\documentclass[11pt,draft,peerreview]{IEEEtran}
\usepackage{amssymb}
\usepackage{amsmath,amsthm,amssymb}
\usepackage{graphicx}
\usepackage{color}
\usepackage{algorithm}  
\usepackage{enumerate}
\usepackage[noend]{algorithmic}
\allowdisplaybreaks
\begin{document}
\raggedbottom
\sloppy
\newtheorem{claim}{Claim}
\newtheorem{corollary}{Corollary}
\newtheorem{definition}{Definition}
\newtheorem{example}{Example}
\newtheorem{exercise}{Exercise}
\newtheorem{fact}{Fact}
\newtheorem{lemma}{Lemma}
\newtheorem{note}{Note}
\newtheorem{obs}{Observation}
\newtheorem{problem}{Problem}
\newtheorem{property}{Property}
\newtheorem{proposition}{Proposition}
\newtheorem{question}{Question}
\newtheorem{ru}{Rule}
\newtheorem{solution}{Solution}
\newtheorem{theorem}{Theorem}
\newenvironment{remark}[1]{\textbf{Remark: }}

\newcommand{\A}{{\bf a}}
\newcommand{\CC}{{\cal C}}
\newcommand{\B}{{\bf b}}
\newcommand{\I}{{\bf i}}
\newcommand{\p}{{\bf p}}
\newcommand{\bh}{{\bf h}}
\newcommand{\q}{{\bf q}}
\newcommand{\x}{{\bf x}}
\newcommand{\y}{{\bf y}}
\newcommand{\z}{{\bf z}}
\newcommand{\M}{{\bf M}}
\newcommand{\N}{{\bf N}}
\newcommand{\X}{{\bf X}}
\def\r {{\bf r}}
\newcommand{\bP}{{\bf P}}
\newcommand{\Q}{{\bf Q}}
\newcommand{\U}{{\bf u}}
\def\d {{\tt d}}
\def\D {{\tt D}}
\def\V {{\tt V}}
\def\Av {{\tt Avg}}
\def\restr{\downharpoonright}
\newcommand{\comment}[1]{\ $[\![${\normalsize #1}$]\!]$ \ }

\newcommand{\restrict}{{\mathbin{\vert\mkern-0.5mu\grave{}}}}

\newcommand{\glb}{{\tt glb}}
\newcommand{\lub}{{\tt lub}}

\newcommand{\LL}{{\cal L}}
\newcommand{\PP}{{\cal P}}
\newcommand{\pv}{{{\bf p}=(p_1,\ldots , p_n)}}

\newcommand{\qv}{{{\bf q}=(q_1,\ldots , q_n)}}
\newcommand{\Inf}{{\bf t}}
\newcommand{\Sup}{{\bf s}}
\newcommand{\inff}{{\tt t}}
\newcommand{\supp}{{\tt s}}
\def\H {{\cal H}}
\def\l {{\bf l}}

\allowdisplaybreaks
\newcommand{\g}{{\bf g}}
\newcommand{\s}{{sup}}
\def\i {{inf}}
\def\r {{\bf r}}
\def\t {{\bf t}}
\def\m{{\bf m}}
\def\half{{\rm half}}
\newcommand{\pcoin}{{coin of bias $p$ }}
\newcommand{\ptree}{{tree of bias $p$ }}
\newcommand{\remove}[1]{}

\newcommand{\todo}[1]{{\tiny\color{red}$\tiny\circ$}%
{\marginpar{\flushleft\\scriptsizes\sf\color{red}$\bullet$#1}}}

\thispagestyle{plain}

\title{How to Find  a Joint Probability Distribution of \\Minimum Entropy (almost) given the  Marginals}
\author{Ferdinando Cicalese, Luisa Gargano, \IEEEmembership{Member, IEEE} and Ugo Vaccaro, \IEEEmembership{Member, IEEE}%
\thanks{F. Cicalese is with the Dipartimento di Informatica,  Universit\`a di Verona, Verona, 
     Italy (email: ferdinando.cicalese@univr.it), L. Gargano is with the Dipartimento di Informatica,  Universit\`a di Salerno,
		Fisciano (SA), Italy (email: lgargano@unisa.it), 
 and U. Vaccaro is with the Dipartimento di Informatica,  Universit\`a di Salerno,
		Fisciano (SA), Italy (email: uvaccaro@unisa.it).	}}
\maketitle
\begin{abstract}
Given two discrete random variables $X$ and $Y$, with probability distributions
$\p=(p_1, \ldots , p_n)$ and $\q=(q_1, \ldots , q_m)$, respectively,
denote by $\CC(\p, \q)$ 
the set of all 
couplings of $\p$ and $\q$, that is, the set of  all  bivariate probability distributions 
that have $\p$ and $\q$ as marginals.
In this paper, we study  the problem of
finding the joint probability distribution in $\CC(\p, \q)$ of
\emph{minimum entropy} (equivalently, the joint probability distribution
that \emph{maximizes} the mutual information between $X$ and $Y$),
and we discuss several situations where the need for
this kind of optimization  naturally arises. 
Since the optimization problem is known to be NP-hard,
we give an efficient algorithm to find a joint probability distribution in $\CC(\p, \q)$
with   entropy  exceeding the minimum possible by at most 1, 
thus providing an approximation algorithm with additive approximation
factor of 1. 

Leveraging on this algorithm, we extend our result to
the problem of finding a minimum--entropy joint distribution
of arbitrary $k\geq 2$ discrete random variables $X_1, \ldots , X_k$, consistent
with the known $k$ marginal distributions of $X_1, \ldots , X_k$.
In this case, our approximation algorithm has an additive approximation
factor of $\log k$. 

 We also discuss some related  applications  of our findings.
\end{abstract}
\section{Introduction and Motivations}
Inferring  an unknown joint distribution of two random variables (r.v.),
 when only their  marginals are given, is an old problem 
in the area of probabilistic inference. The problem goes back at least to Hoeffding \cite{Hoeff}
and Frechet \cite{frechet},
who studied the question of identifying the extremal joint distribution of
r.v. $X$ and $Y$ that maximizes (resp., minimizes) their correlation, 
given the    marginal distributions of $X$ and $Y$.
We refer the reader to \cite{benes,cuadras,Aglio,Lin} for a (partial) account of the
 vast literature in the area and the many  applications in the pure and applied sciences.

In this paper, we consider the following  case  of the general  problem  described above. Let $X$ and $Y$ be 
two discrete r.v., distributed according to $\p=(p_1, \ldots , p_n)$ and $\q=(q_1, \ldots , q_m)$, respectively. 
We seek a \emph{minimum-entropy} joint probability distribution of $X$ and $Y$, whose   marginals are equal  to 
$\p$ and $\q$. This problem arises in many situations. For instance, the authors of \cite{ko+}
consider the  important question  of identifying the correct causal direction between two arbitrary
r.v. $X$ and $Y$, that is,  they want to discover whether it is  the case that $X$ 
causes $Y$ or  it is  $Y$  that causes $Y$.  In general,
$X$  causes $Y$ if there exists an exogenous  r.v. $E$ and a deterministic function $f$ such that
$Y=f(X,E)$. In order to identify the correct causal direction (i.e., either from $X$ to $Y$ or 
from $Y$ to $X$), the authors of \cite{ko+} make the reasonable postulate that the entropy of the  exogenous  r.v. $E$ is {small} in the 
\emph{true} causal direction, and empirically validate this assumption.
Additionally, they prove the interesting fact that 
the problem of finding the exogenous variable $E$  with minimum entropy
 is equivalent to the problem of finding the minimum-entropy joint
distribution of properly defined random variables, \emph{given} (i.e., fixed) their marginal distributions.
This is exactly    the problem we consider in this paper.
The authors of \cite{ko+} also observe that the latter optimization 
problem is NP-hard (due to results of \cite{KSS15,V12}), and evaluate experimentally 
a greedy  approximation algorithm  to find the minimum-entropy joint
distribution, given the marginals. No proved  performance guarantee is
given in \cite{ko+} for that  algorithm. In this paper, we give a
(different) greedy algorithm and   we  prove that it
returns a correct  joint probability distribution (i.e., with 
the prescribed marginals)
with   entropy  exceeding the \emph{minimum possible by at most of 1}.
Subsequently, in Section \ref{sec:more} we extend our algorithm to the case of more than two
random variables.  More precisely, 
we consider 
the problem of finding a minimum--entropy joint distribution
of arbitrary $k\geq 2$ discrete random variables $X_1, \ldots , X_k$, consistent
with the known $k$ marginal distributions of $X_1, \ldots , X_k$.
In this case, our approximation algorithm has an additive approximation
factor of $\log k$. 

\smallskip
Another work  that considers the problem of finding the 
minimum-entropy joint  distribution 
of two r.v. $X$ and $Y$, given the marginals of $X$ and $Y$, is the paper 
\cite{V12}. There, the author introduces   a pseudo-metric $\D(\cdot, \cdot)$ among discrete 
probability distributions in the following way:
given arbitrary $\p=(p_1, \ldots, p_n)$ 
and $\q=(q_1, \ldots , q_m)$, $m\leq n$, the distance $\D(\p,\q) $
among $\p$ and $\q$  is defined as the quantity
 $\D(\p,\q)=2W(\p,\q)-H(\p)-H(\q)$,
where $W(\p,\q)$ is the \emph{minimum} entropy of a bivariate
probability distribution that has $\p$ and $\q$ as marginals,
and $H$ denotes the Shannon entropy. This metric is applied  in \cite{V12}  to the problem of order-reduction
of stochastic processes.
The author of \cite{V12} observes that the problem of computing
$W(\p,\q)$ is NP-hard and proposes another different  greedy algorithm for its
computation, based on some analogy with the problem of Bin Packing
with overstuffing. Again, no performance guarantee is given in
\cite{V12} for the proposed algorithm. Our result directly implies
that we can compute the value of $\D(\p,\q)$, for \emph{arbitrary} $\p$ and $\q$,
with an additive error of at most $1$.\footnote{We remark that in \cite{CGVold}
we considered   the different problem of computing  the
probability distributions $\q^*$  that \emph{minimizes} $\D(\p,\q)$,
given $\p$.}

\smallskip
There are many other problems that require the computation of the 
{minimum-entropy} joint probability distribution of two 
random variables, whose   marginals are equal  to 
$\p$ and $\q$. We shall limit ourselves
to discuss a few additional examples, postponing a more
complete examination in a future  version of the paper.
To this purpose, let us write the joint entropy of two r.v. $X$ and $Y$,
distributed according to $\p$ and $\q$, respectively,  as
$H(XY)=H(X)+H(Y)-I(X;Y)$, where $I(X;Y)$ is the mutual information
between $X$ and $Y$. Then,  one sees that our original problem 
can  be equivalently stated as the determination  of  a joint probability distribution of
$X$ and $Y$ (having   given marginals $\p$ and $\q$) that \emph{maximizes} the mutual information
$I(X;Y)$. In the paper \cite{KSS15} this maximal mutual information is interpreted,
in agreement with Renyi axioms for a \emph{bona fide} dependence measure \cite{renyi},
as a measure of the \emph{largest possible} dependence of the two r.v. $X$ and $Y$.
Since the problem of its exact computation is obviously NP-hard, our result
implies  an approximation algorithm for it.
Another situation  where  the need to maximize the mutual information 
between two r.v. (with fixed probability distributions) naturally   arises, is in the area of medical imaging
 \cite{pluim, wells}.
Finally, our problem could also be seen as a kind of ``channel-synthesis'' problem, where it is 
given pair of r.v. $(X,Y)$, and the goal is to construct a memoryless channel
that maximizes the mutual information $I(X;Y)$ between $X$ and $Y$.

\section{Mathematical Preliminaries}

We start by recalling  a few  notions of majorization theory \cite{MO}
that  are relevant to our context.

\begin{definition}\label{defmaj} 
Given two probability distributions
$\A=(a_1, \ldots ,a_n)$ and $\B=(b_1, \ldots , b_n)$ with $a_1\geq \ldots \geq a_n\geq 0$ and 
$b_1\geq \ldots \geq b_n\geq 0$, $\sum_{i=1}^na_i=\sum_{i=1}^nb_i=1$, we say that $\A$ is 
{\em majorized} by $\B$, and write  $\A \preceq \B$,
if and only if
$\sum_{k=1}^i a_k\leq \sum_{k=1}^i b_k, \quad\mbox{\rm for all }\  i=1,\ldots , n.$
\end{definition}
We assume that \emph{all the probabilities distributions 
we deal with have been
ordered in non-increasing order}. This assumption does not affect our results, since the quantities we compute
(i.e., entropies) are invariant with respect to permutations of the components of
the involved probability distributions. We also  use the majorization 
relationship between vectors of unequal lengths, by properly padding the shorter
one with the appropriate number of $0$'s at the end.
The majorization relation $\preceq$ is a partial ordering on the set
$$\PP_n=\{(p_1,\ldots, p_n)\ :  \sum_{i=1}^n p_i=1, \ p_1\geq \ldots \geq p_n\geq 0\}$$  
of all ordered probability vectors of $n$ elements, that is, for each $\x,\y,\z\in\PP_n$
it holds that 
\begin{itemize}
\item[1)] $\x\preceq \x$;
\item[2)] $\x\preceq \y$ and $\y\preceq \z$  implies $\x\preceq \z$;
\item[3)]  $\x\preceq \y$ and $\y\preceq \x$  implies $\x =\y$.
\end{itemize}

It turns out that   that the partially ordered set $(\PP_n,\preceq)$ is indeed a \emph{lattice} \cite{CV},\footnote{The same result was independently 
rediscovered in \cite{cuff}, see also \cite{harr} for a different   proof.} 
i.e., for all $\x,\y  \in \PP_n$
there exists
a unique
{\em least upper bound}  $\x\lor \y$ and 
 a unique {\em greatest lower bound}  $\x \land \y$.
We recall that  the least upper bound
 $\x \lor \y$ is the vector in $\PP_n$ such that:  
$$\x\preceq \x \lor \y, \ \y\preceq \x \lor \y, \ \mbox{ and 
for all } \ \z\in \PP_n  \ \mbox{ for which } \
\x \preceq \z,\ \y \preceq \z\ \mbox{ it holds that  } \ \x\lor \y\preceq \z.$$
Analogously, the greatest lower bound $\x \land \y$ is the vector in $\PP_n$ such that: 
$$ \x \land \y\preceq \x,  \x \land \y\preceq \y, \ \mbox{ and 
for all } \  \z\in \PP_n \ \mbox{ for which } \ 
\z \preceq \x,\ \z \preceq \y\ \mbox{ it holds that  } \ \z\preceq \x\land \y.$$
In the  paper \cite{CV} the authors also gave a simple and efficient algorithm to 
compute $\x \lor \y$ and $\x\land \y$, given arbitrary vectors $\x,\y\in \PP_n$. 
Due to the important  role it will play in our main result, 
we recall how to  compute  the greatest lower bound. 
\begin{fact}{\rm\cite{CV}} \label{glb}
Let $\x = (x_1, \dots, x_n), \y = (y_1, \dots, y_n) \in \PP_n$ and let $\z = (z_1, \dots, z_n) = \x \land \y$. 
Then, {$z_1=\min\{p_1, q_1\}$ and 
for each $i = 2, \dots, n,$ it holds that $$z_i = \min \Bigl \{\sum_{j=1}^i p_j, \sum_{j=1}^i q_j\Bigr\} 
- \sum_{j=1}^{i-1} z_{j}.$$  }
\end{fact}

\noindent
We also remind   the important  Schur-concavity property  of the entropy function \cite{MO}:

\emph{For any  $\x,\y\in\PP_n$, 
 $\x\preceq \y$ implies that    $H(\x)\geq H(\y)$, with equality 
    if and only if   $ \x=\y$. }
		
		\medskip
A notable strengthening  of above fact has been proved  in \cite{HV}. There, the authors prove that 
$\x\preceq \y$ implies $H(\x)\geq H(\y)+D(\y||\x)$, where $D(\y||\x)$ is the relative entropy
between $\x$ and $\y$.

We also need the  concept of \emph{aggregation} \cite{V12,CGVold} and a result from \cite{CGVold},
whose proof is repetead here to make the paper self-contained.
Given $\p=(p_1, \ldots , p_n)\in \PP_n $, 
we say that $\q=(q_1, \ldots , q_m)\in \PP_m$ is an \emph{aggregation} 
of $\p$ if there is a partition of $\{1, \ldots , n\}$ into disjoint sets $I_1, \ldots , I_m$
such that $q_j=\sum_{i\in I_j}p_i$, for $j=1, \ldots m$.

\begin{lemma}{\rm\cite{CGVold}}\label{pprecq}
Let $\q\in \PP_m $ be \emph{any} aggregation of $\p\in \PP_n$.  Then it holds that
$\p\preceq \q$.
\end{lemma}
\begin{IEEEproof} 
We shall prove  by induction on $i$ that $\sum_{k=1}^{i}q_k\geq \sum_{k=1}^{i}p_k$.
Because  $\q$ is an aggregation of $\p$, we know that there exists    $I_j\subseteq \{1, \ldots , n\}$
such that $1\in I_j$. This implies that  $q_1\geq q_j\geq p_1$. Let us suppose that 
$\sum_{k=1}^{i-1}q_k\geq \sum_{k=1}^{i-1}p_k$. 
If there exist  indices $j\geq i$ and $\ell\leq i$ such that $\ell\in I_j$,
then $q_i\geq q_j\geq p_\ell\geq p_i$,  implying 
$\sum_{k=1}^{i}q_k\geq \sum_{k=1}^{i}p_k$.  
Should it be otherwise,  for each $j\geq i$ and $\ell\leq i$ it holds that $\ell\not \in I_j$. Therefore,
$\{1, \ldots ,i\}\subseteq I_1\cup \ldots \cup I_{i-1}$. This immediately gives
$\sum_{k=1}^{i-1}q_k\geq \sum_{k=1}^{i}p_k$, from which we 
get $\sum_{k=1}^{i}q_k\geq \sum_{k=1}^{i}p_k$.
\end{IEEEproof}

\bigskip

Let us now discuss some  consequences of above framework.
Given two discrete random variables $X$ and $Y$, with probability distributions
$\p=(p_1, \ldots , p_n)$ and $\q=(q_1, \ldots , q_m)$, respectively,
denote by $\CC(\p, \q)$ 
the set of all joint distributions of $X$ and $Y$ 
that have $\p$ and $\q$ as marginals (in the literature, elements of $\CC(\p, \q)$ are often called \emph{couplings} of 
$\p$ and $\q$). For our purposes,
each element in $\CC(\p, \q)$ can be seen as $n\times m$ matrix $M=[m_{ij}]\in \mathbb{R}^{n\times m}$
such that its row-sums give the elements of $\p$ and its column-sums give the elements of $\q$, that is, 
\begin{equation}\label{CC}
\CC(\p, \q)=\Bigl\{\mathbf{M}=[m_{ij}]: \sum_{j}m_{ij}=p_i, \sum_{i}m_{ij}=q_j\Bigr\}.
\end{equation}
Now, for any  $\mathbf{M}\in \CC(\p, \q)$, let us write its elements in a $1\times mn$ vector
$\mathbf{m}\in \PP_{mn}$, with its components ordered in non-increasing fashion.
From (\ref{CC}) we obtain  that both  $\p$ and $\q$ are  {aggregations}
of \emph{each } $\mathbf{m}\in \PP_{mn}$  obtained from some  $\mathbf{M}\in \CC(\p, \q)$. 
By Lemma \ref{pprecq}, we get that\footnote{Recall that we 
use the majorization 
relationship between vectors of unequal lenghts, by properly padding the shorter
one with the appropriate number of $0$'s at the end. This trick does not affect our 
subsequent results, since 
we use the customary assumption that $0\log\frac{1}{0}=0$.}
\begin{equation}\label{m<peq}
\mathbf{m}\preceq \p \quad \mbox{and} \quad \mathbf{m}\preceq \q.
\end{equation}
Recalling the definition and properties of the greatest lower bound
of two vectors in $\PP_{mn}$, we also obtain
\begin{equation}\label{m<pandq}
\mathbf{m}\preceq\p\land \q.
\end{equation}
From (\ref{m<pandq}), and the Schur-concavity of the Shannon entropy, we also obtain that
$$H(\mathbf{m})\geq H(\p\land \q).$$
Since, obviously, the entropy of $H(\mathbf{m})$ is equal to the entropy $H(\mathbf{M})$,
where $\mathbf{M}$ is the matrix in $\CC(\p, \q)$ from which the vector $\mathbf{m}$
was obtained, we get the following  important  result (for us).
\begin{lemma}\label{lemma:HM>Hpandq}
\emph{For any} $\p$ and $\q$, and  $\mathbf{M}\in \CC(\p, \q)$, it holds that 
\begin{equation}\label{eq:HM>Hpandq}
H(\mathbf{M})\geq H(\p\land \q).
\end{equation}
\end{lemma}
Lemma \ref{lemma:HM>Hpandq} is one of the  key results towards our approximation algorithm
to find an  element $\mathbf{M}\in \CC(\p, \q)$ with entropy
$H(\mathbf{M})\leq OPT+1$, where 
$OPT=\min_{\mathbf{N}\in \CC(\p, \q)} H(\mathbf{N}).$

Before describing our algorithm, let us illustrate some interesting 
consequences of Lemma \ref{lemma:HM>Hpandq}.
It is well known that for any joint distribution of the two r.v. $X$ and $Y$ it
holds that 
\begin{equation}\label{trivial}
H(XY)\geq \max\{H(X), H(Y)\},
\end{equation}
or, \emph{equivalently}, 
for any $\mathbf{M}\in \CC(\p, \q)$ it
holds that 
$$H(\mathbf{M})\geq \max\{H(\p), H(\q)\}.$$
Lemma \ref{lemma:HM>Hpandq}
strengthens the  lower bound 
(\ref{trivial}). Indeed, since, by definition, it holds that
$\p\land \q\preceq \p$ and $\p\land \q\preceq \q$, by the Schur-concavity of
the entropy function and Lemma \ref{lemma:HM>Hpandq} we get  the (improved) lower bound
\begin{equation}\label{improved}
H(\mathbf{M})\geq H(\p\land \q)\geq \max\{H(\p), H(\q)\}.
\end{equation}
Inequality (\ref{improved}) also allows us to improve on  the 
classical  upper bound on the mutual information given by $I(X;Y)\leq \min\{H(X), H(Y)\},$
since (\ref{improved}) implies
\begin{equation}\label{improved2}
I(X;Y)\leq H(\p)+H(\q)-H(\p\land \q)\leq \min\{H(X), H(Y)\}.
\end{equation}
The new bounds are \emph{strictly} better than the usual ones, whenever
$\p\not\preceq\q$ and $\q\not\preceq\p$.
Technically, one could 
improve them even more, by using the 
inequality $H(\x)\geq H(\y)+D(\y||\x)$, whenever $\x\preceq \y$
 \cite{HV}. However,   in this paper we just need what we can get from 
 the inequality  $H(\x)\geq H(\y)$, whenever  $\x\preceq \y$ holds.
 
{Inequalities (\ref{improved}) and (\ref{improved2}) could be useful also in other contexts,  when
one needs to bound the joint entropy (or the mutual information) of two
r.v. $X$ and $Y$,
and the only available knowledge is given by  the marginal distributions of $X$ and $Y$
(and not their  joint distribution).
}

\section{Approximating \ \ $OPT=\min_{\mathbf{N}\in \CC(\p, \q)} H(\mathbf{N}).$}
In this section we present an algorithm that, having in input distributions
$\p$ and $\q$, constructs an $\mathbf{M}\in \CC(\p, \q)$ such that 
\begin{equation}\label{eq:HM<Hpandq+1}
H(\mathbf{M})\leq H(\p\land \q)+1.
\end{equation}
 Lemma \ref{lemma:HM>Hpandq} will imply  that
$$ H(\mathbf{M})\leq \min_{\mathbf{N}\in \CC(\p, \q)} H(\mathbf{N})+1. $$
We need to introduce some additional notations and state some properties 
which will be used in the description of  our algorithm.

\begin{definition} \label{defi:inversions}
Let $\p = (p_1, \dots, p_n)$ and $\q = (q_1, \dots, q_n)$ be two probability distributions in $\PP_n$. 
We assume that   
for the maximum $i \in \{1, \dots, n\}$ such that $p_i \neq q_i$---if it exists---it holds that $p_i > q_i.$\footnote{Notice that up to swapping the role 
of $\p$ and $\q$, the definition applies to any pair of {\em distinct} distributions.}
Let $k$ be the minimum integer such that there are  
indices $i_0 = n+1 > i_1 > i_2 > \cdots > i_k = 1$ satisfying the following conditions for each $s = 1, \dots, k$:
\begin{itemize}
\item  if $s$ is odd, then $i_s$ is the minimum index smaller than $i_{s-1}$ such that $\sum_{k=i}^n p_k \geq \sum_{k=i}^n q_k$ 
holds for each $i = i_s, i_s+1, \dots, i_{s-1}-1$;
\item  if $s$ is {even}, then $i_s$ is the minimum index smaller than $i_{s-1}$ such that $\sum_{k=i}^n p_k {\leq} \sum_{k=i}^n q_k$ 
holds for each $i = i_s, i_s+1, \dots, i_{s-1}-1$.
\end{itemize}

\medskip
We refer to the integers $i_0, i_1, \dots, i_k$ as the {\em inversion points} of $\p$ and  $\q$.\footnote{If $\p = \q$ 
then we have $k = 1$ and $i_1 = 1.$}

\end{definition}

%

\begin{fact} \label{fact:p-q-z}
Let $\p$ and $\q$ be two probability distributions in 
$\PP_n$ and $i_0 = n+1 > i_1 \cdots > i_k = 1$ be their inversion points. Let $\z =  \p \wedge \q.$
Then the following relationships hold:
\begin{enumerate}
\item for each odd $s \in \{1, \dots, k\}$ and $i \in \{i_s, \dots, i_{s-1}-1\}$
\begin{equation} \label{eq:z-p-q-suffix-sums1}
\sum_{k = i}^n z_k = \sum_{k=i}^n p_k
\end{equation}
\item for each even $s \in \{1, \dots, k\}$ and $i\in \{i_s, \dots, i_{s-1}-1\}$
\begin{equation} \label{eq:z-p-q-suffix-sums2}
\sum_{k = i}^n z_k = \sum_{k=i}^n q_i
\end{equation}
\item for each odd $s \in \{1, \dots, k\}$ and $i\in \{i_s, \dots, i_{s-1}-2\}$ we have $z_i = p_i$
\item for each even $s \in \{1, \dots, k\}$ and $i\in \{i_s,  \dots, i_{s-1}-2\}$ we have $z_i = q_i$
\item for each odd $s \in \{0, \dots, k-1\}$
\begin{equation} \label{eq:z-p-q-inv-odd1}
z_{i_s-1} = q_{i_s-1} - \left( \sum_{k=i_s}^n p_k - \sum_{k=i_s}^n q_k \right) \geq p_{i_s-1}
\end{equation}
\item for each even $s \in \{0, \dots, k-1\}$
\begin{equation} \label{eq:z-p-q-inv-odd2}
z_{i_s-1} = p_{i_s-1} - \left( \sum_{k=i_s}^n q_k - \sum_{k=i_s}^n p_k \right) \geq q_{i_s-1}
\end{equation}
\end{enumerate}
\end{fact}
\begin{IEEEproof}
By Fact \ref{glb} it holds that 
$$\sum_{k=1}^{i} z_k = \min \Bigl\{\sum_{k=1}^{i} p_k , \sum_{k=1}^{i} q_k  \Bigr\}.$$
 Equivalently, using 
$\sum_k z_k = \sum_k p_k = \sum_k q_k = 1,$ we have 
that for each $i = 1, \dots, n,$ it holds that 
$$\sum_{k=i}^{n} z_k = \max \Bigl\{\sum_{k=i}^{n} p_k , \sum_{k=i}^{n} q_k  \Bigr\}.$$
This, together with 
the definition of the inversion indices $i_0, \dots, i_k$, imply properties 1)\ and 2).
The remaining properties are easily derived from 1) and 2) by simple algebraic calculations. 
\end{IEEEproof}




\begin{lemma} \label{lemma:pezzettini}
Let $A$ be a multiset of non-negative real numbers and $z$ a positive real number such that $z \geq y$ for each $y \in A.$ 
For any $x \geq 0$ such that $x \leq z + \sum_{y \in A} y$  there exists a subset $Q \subseteq A$ 
and $0 \leq z^{(d)} \leq z$ such that  $$z^{(d)} + \sum_{y \in Q} y  = x.$$ Moreover, $Q$ and $z^{(d)}$ can 
be computed in linear time. 
\end{lemma}
\begin{IEEEproof}
\remove{
The key  is to show that under the hypotheses of the lemma, 
there exists a (possibly empty) subset $Q \subseteq A$ such that $\sum_{y \in Q} y < x$ and 
$z + \sum_{y \in Q} y \geq x$ and that such a $Q$ can be determined in linear time. 
Then, setting $z^{(d)} = x - \sum_{y \in Q} y$ we get the desired result. 

}
If $\sum_{y \in A} y < x,$  we get $Q = A$ and the desired result directly follows from the assumption that $z + \sum_{y \in A} y \geq x.$ Note that 
the condition can be checked in linear time.

Let us now assume that $\sum_{y \in A} y \geq x.$
Let   $y_1, \dots, y_k$ be the elements of $P$. 
Let $i$ be the minimum index such that
$\sum_{j=1}^i y_j \geq x.$ Then setting $Q = \{y_1, \dots, y_{i-1}\}$  (if $i = 1,$ we set $Q = \emptyset$) 
and using the assumption that $z \geq y_i$
we have the desired result. Note that also in this case the index $i$ which determines $Q = \{y_1, \dots, y_{i-1}\},$ can be found in linear time. 
\end{IEEEproof}

\medskip

This lemma is a major technical  tool of our main algorithm. 
We present a procedure implementing the 
construction of the the split of $z$ and the 
set $Q$ in \textbf{Algorithm \ref{algo:Lemma}}. 

\begin{algorithm}[ht!]
\small
{\sc  Min-Entropy-Joint-Distribution}($\p , \q$)\\ 
{\bf Input:} prob.\ distributions $\p = (p_1, \dots, p_n)$ and $\q= (q_1, \dots, q_n)$\\
{\bf Output:} An $n \times n$ matrix $\mathbf{M} = [m_{i\,j}]$ s.t.\ $\sum_{j} m_{i\,j} = p_i$ and $\sum_{i} m_{i\,j} = q_j.$
\begin{algorithmic}[1]
\STATE{{\bf for} $i=1,\dots, n$ and $j=1, \dots, n$  {\bf set} $m_{i \, j} \leftarrow 0$} \label{ini:1}
\STATE{{\bf for} $i=1,\dots, n$  {\bf set} $R[i] \leftarrow 0,\, C[i] \leftarrow 0$} \label{ini:2}
\STATE{{\bf if} $\p \neq \q$, let $i = \max \{j \mid p_j \neq q_j\}$; {\bf if} $p_i < q_i$ {\bf swap} $\p \leftrightarrow \q$} \label{ini:3}
\STATE{Let $i_0 = n+1 > i_1 > i_2 > \cdots > i_k = 1$ be the inversion indices of $\p$ and $\q$  as by Definition \ref{defi:inversions}} \label{ini:4}
\STATE{$\z = (z_1, \dots, z_n) \leftarrow \p \land \q$}
\FOR{$s = 1$ {\bf to} $k$} \label{main:1}
\IF{$s$ is odd} 
\FOR{$j = i_{s-1}-1$ {\bf downto} $i_s$}  \label{int2.1-start}
\STATE{($z_{j}^{(d)},  z_{j}^{(r)}, Q$) $\leftarrow$  {\sc  Lemma\ref{lemma:pezzettini}}($z_j, q_j, R[j+1 \dots i_{s-1}-1]$)}   
\FOR{{\bf each } $\ell \in Q$}   \label{inner-for1-start}
\STATE{$m_{\ell\, j} \leftarrow R[\ell]; \, R[\ell] \leftarrow 0$} \label{m-change1}
\ENDFOR   
\STATE{$m_{j\, j} \leftarrow z^{(d)}_j, \, R[j] \leftarrow z^{(r)}_j$} \label{inner-for1-end} \label{m-change2}
\ENDFOR  \label{int2.1-end}
\IF{$i_s \neq 1$} \label{if1}
\FOR{{\bf each} $\ell \in [i_s\dots i_{s-1}-1]$ s.t. $R[\ell] \neq 0$}  \label{inversion-sums1-start}
\STATE{$m_{\ell\, i_{s}-1} \leftarrow R[\ell]; \, \, R[\ell] \leftarrow 0$} \label{inversion-sums1-end} \label{m-change3}
\ENDFOR
\ENDIF   
\ELSE   \label{int2.1-if}
\FOR{$j = i_{s-1}-1$ {\bf downto} $i_s$} \label{int2.2-start}
\STATE{($z_{j}^{(d)},  z_{j}^{(r)}, Q$) $\leftarrow$  {\sc  Lemma\ref{lemma:pezzettini}}($z_j, p_j, C[j+1 \dots i_{s-1}-1]$)}   
\FOR{{\bf each } $\ell \in Q$} \label{inner-for2-start}
\STATE{$m_{j \, \ell} \leftarrow C[\ell], \, C[\ell] \leftarrow 0$} \label{m-change4}
\ENDFOR 
\STATE{$m_{j\, j} \leftarrow z^{(d)}_j, \, C[j] \leftarrow z^{(r)}_j$} \label{inner-for2-end} \label{m-change5}
\ENDFOR  \label{int2.2-end}
\IF{$i_s \neq 1$}   \label{if2}
\FOR{{\bf each} $\ell \in [i_s\dots i_{s-1}-1]$ s.t. $C[\ell] \neq 0$}  \label{inversion-sums2-start}
\STATE{$m_{i_{s}-1 \, \ell} \leftarrow C[\ell]; \, \, C[\ell] \leftarrow 0$} \label{inversion-sums2-end} \label{m-change6}
\ENDFOR
\ENDIF \label{int2.2-if}
\ENDIF
\ENDFOR
\end{algorithmic}
\caption{The Min Entropy Joint Distribution Algorithm}
\label{algo:Algonew}
\end{algorithm}

\begin{algorithm}[ht!]
\small
{\sc  Lemma\ref{lemma:pezzettini}}($z, x, A[i\dots j]$)\\ 
{\bf Input:} reals $z > 0, \, x \geq 0,$ {  and $A[i \dots j] \geq 0$ 
s.t.  $\sum_{k} A[k] + x \geq z$ } \\
{\bf Output:} $z^{(d)}, z^{(r)} \geq 0,$ and $Q \subseteq \{i, i+1, \dots, j\}$ s.t. $z^{(d)} + z^{(r)} = z,$ and $z^{(d)} + \sum_{\ell \in Q} A[\ell] = x.$

\begin{algorithmic}[1]
\STATE{$k \leftarrow i, \, Q \leftarrow \emptyset, \, sum \leftarrow 0$}  
\WHILE{$k \leq j$ {\bf and} $sum + A[k] < x$}
\STATE{$Q \leftarrow Q \cup \{k\}, \, sum \leftarrow sum + A[k], \, k \leftarrow k+1$} 
\ENDWHILE
\STATE{$z^{(d)} \leftarrow x - sum, \, z^{(r)} \leftarrow z - z^{(d)}$}
\STATE{{\bf return} ($z^{(d)}, z^{(r)}, Q$)}
\end{algorithmic}
\caption{The procedure implementing Lemma \ref{lemma:pezzettini}}
\label{algo:Lemma}
\end{algorithm}

By padding the probability distributions with the appropriate
number of $0$'s, we can assume that both  $\p,\q\in \PP_n$. We are now ready to present our main algorithm.
 The pseudocode is presented is given below (\textbf{Algorithm 1}).
An informal description of it, that gives also the intuition behind its functioning, is 
presented in subsection \ref{informal}. 

The following theorem shows the correctness of \textbf{Algorithm \ref{algo:Algonew}}. It relies on 
a sequence of technical results, Lemmas \ref{lemmata:prop1} and \ref{lemmata:prop3} 
and Corollaries \ref{lemmata:prop2} and \ref{lemmata:prop4},  whose statements are deferred to the end of this section.

\begin{theorem} \label{theo:correctness}
For any pair of probability distributions $\p, \q \in \PP_n$ the output of Algorithm 
\ref{algo:Algonew} is an $n \times n$ matrix $\mathbf{M} = [m_{i\, j}] \in \CC(\p, \q)$ i.e., such that 
$\sum_{j} m_{i\, j} = p_i$ and $\sum_{i} m_{i\, j} = q_j.$  
\end{theorem}
\begin{IEEEproof}
Let $k$ be the value of $s$ in the last iteration of the line 6. {\bf for}-loop of the
algorithm. The desired result directly follows by Corollaries \ref{lemmata:prop2} and 
\ref{lemmata:prop4} according to whether $k$ is odd or even, respectively.
\end{IEEEproof}

\medskip
We now prove our main result.

\begin{theorem}
For any $\p, \q \in \PP_n$, Algorithm \ref{algo:Algonew} outputs {\em in 
polynomial time} an $\mathbf{M}\in \CC(\p, \q)$ such that 
\begin{equation} \label{eq:maintheo}
H(\mathbf{M})\leq H(\p\land \q)+1. 
\end{equation}
\end{theorem}
\begin{IEEEproof}
%
It is not hard to see that the non-zero entry of the matrix  $\mathbf{M}$  are all fixed in 
lines \ref{inner-for1-end} and \ref{inner-for2-end}---in fact, for the assignments  
in lines \ref{inversion-sums1-end} and \ref{inversion-sums2-end} the algorithm uses values
stored in $R$ or $C$ which were fixed at some point earlier in lines \ref{inner-for1-end} and \ref{inner-for2-end}.
Therefore, all the final non-zero entries of $\mathbf{M}$ can be partitioned into $n$ pairs $z_j^{(r)}, z_j^{(d)}$ with 
$z_j^{(r)} + z_j^{(d)} = z_j$ for $j = 1, \dots, n$. By using the standard assumption $0 \log \frac{1}{0} = 0$ and 
applying Jensen inequality we have
\begin{eqnarray*}
H(\mathbf{M}) &=& 
\sum_{j=1}^n z_j^{(r)} \log \frac{1}{z_j^{(r)}} + z_j^{(d)} \log \frac{1}{z_j^{(d)}} \\
&\leq& \sum_{j=1}^n \frac{z_j}{2} \log \frac{2}{z_j} = H(\z) + 1
\end{eqnarray*}
which concludes the proof of the additive approximation guarantee of Algorithm \ref{algo:Algonew}.
Moreover, one can see that Algorithm \ref{algo:Algonew} can be implemented so to run in $O(n^2)$ time.
For the time complexity of the algorithm we observe the following easily verifiable fact:
\begin{itemize}
\item the initialization in lines \ref{ini:1}-\ref{ini:2} takes $O(n^2)$;
\item the condition in line \ref{ini:3} can be easily verified in $O(n)$ which is also the complexity of 
swapping $\p$ with $\q$, if needed;
\item the computation of the inversion points of $\p$ and $\q$ in line \ref{ini:4} can be 
performed in $O(n)$ following Definition \ref{defi:inversions}, once the suffix sums 
$\sum_{j=k}^n p_j, \, \sum_{j=k}^n q_j$ ($k = 1, \dots, n$) have been precomputed (also doable in $O(n)$);
\item the vector $z = \p \land \q$ can be computed in $O(n)$, e.g., based on the precomputed suffix sums; 
\item in the main body of the algorithm, the most expensive parts are the 
calls to the procedure {\bf Lemma3}, and the {\bf for}-loops in lines \ref{inner-for1-start}, 
\ref{inner-for2-start}, \ref{inversion-sums1-start}, and \ref{inversion-sums2-start}. All these take $O(n)$ and 
it is not hard to see that they are executed at most $O(n)$ times (once per component of $\z$). Therefore, 
the main body of the algorithm in lines \ref{main:1}-\ref{int2.2-if} takes $O(n^2)$.
\end{itemize}
Therefore we can conclude that the time complexity of Algorithm \ref{algo:Algonew} is polynomial in $O(n^2).$
\end{IEEEproof}

\subsection{The analysis of correctness of Algorithm \ref{algo:Algonew}: technical lemmas} \label{sec:lemmata}
In this section we state four technical lemmas we used for the analysis of Algorithm \ref{algo:Algonew} which 
leads to Theorem \ref{theo:correctness}. 
In the Appendix, we give  a numerical example 
of an execution of \textbf{Algorithm 1}.

\begin{lemma} \label{lemmata:prop1}
{At the end of each iteration of the {\bf for}-loop of lines \ref{int2.1-start}-\ref{int2.1-end} in Algorithm  \ref{algo:Algonew} 
($s = 1, \dots, k,$ and $i_s\leq j <i_{s-1}$)}
we have {(i) $m_{\ell\, c} = 0$ for each $\ell, c$ such that $\min\{\ell, c\} < j$; 
and (ii)} for each $j' = j, \dots, i_{s-1}-1$
\begin{equation}\label{prop1:eq}
\sum_{k \geq i_s} m_{k \, j'} = q_{j'}, \qquad \mbox{ and } \qquad R[j'] + \sum_{k \geq i_s} m_{j'\, k} = p_{j'}. 
\end{equation}
\end{lemma}
\begin{IEEEproof}
For (i) we observe that, before the first iteration ($j = i_{s-1}-1$) the condition holds (by line \ref{ini:1}, when $s=1$, and  
by Corollary \ref{lemmata:prop4} for odd  $s > 1$). 
Then, within each iteration of the {\bf for}-loop values $m_{\ell \, c}$ only change 
in lines \ref{m-change2}, where (i) is clearly preserved, and in line \ref{m-change1} where, as the result of call to Algorithm \ref{algo:Lemma},
we have $\ell \in Q \subseteq \{j+1, \dots, i_{s-1}-1\},$ which again preserves  (i).

We now prove (ii) by induction on the value of $j$.
First we observe that at the beginning of the first iteration of the {\bf for}-loop (lines \ref{int2.1-start}-\ref{int2.1-end}), i.e., 
for $j= i_{s-1}-1,$ it holds that 
\begin{equation} \label{lemma9:eq1}
\sum_{k \geq i_{s-1}} m_{i_{s-1}-1\, k}  =  \sum_{k = i_{s-1}}^n q_k - \sum_{k = i_{s-1}}^n p_k = p_{i_{s-1}-1} - z_{i_{s-1}-1}.
\end{equation}
This is true when $s = 1,$ since in this case we have $i_{s-1} = n+1,$ hence the two sums in the middle term are both 0; the first term is 
0 since no term in $M$ has been fixed yet, and the last term is also 0, since $p_n = z_n$ by assumption.
The equation is also true for each odd $s > 1$ by Corollary \ref{lemmata:prop4}.
Moreover, at the beginning of the first iteration ($j = i_{s-1}-1$) it holds that $R[\ell] = 0$ for each $\ell = 1, \dots, n.$ This is true for $s = 1$ 
because of the initial setting in line 2. For $s > 1$ the property holds since any $R[\ell]$ is only assigned non-zero value within 
the {\bf for}-loop (lines \ref{int2.1-start}-\ref{int2.1-end}) and unless the algorithm stops any non-zero $R[\ell]$ is 
zeroed again immediately after the {\bf for}-loop, 
in lines \ref{inversion-sums1-start}-\ref{inversion-sums1-end} unless the
exit condition $i_s = 1$ is verified which means that $s = k$ and the algorithm terminates immediately after.  

When Algorithm \ref{algo:Algonew} enters the {\bf for}-loop at lines \ref{int2.1-start}-\ref{int2.1-end}, 
$s$ is odd. Then, by point 6.\ of Fact \ref{fact:p-q-z} and (\ref{lemma9:eq1}) it holds that 
\begin{equation} \label{lemma9:eq2}
q_{i_{s-1}-1} \leq z_{i_{s-1}-1} = p_{i_{s-1}-1} - \sum_{k \geq i_{s-1}} m_{i_{s-1}-1\, k}.
\end{equation} 
This implies that for $j = i_{s-1}-1$ the values $z_j, q_j$ together with the values in $R[j+1 \dots i_{s-1}-1]$ 
satisfy the hypotheses of Lemma \ref{lemma:pezzettini}. Hence, 
the call in line 7 to  Algorithm \ref{algo:Lemma} (implementing the construction in the proof of Lemma \ref{lemma:pezzettini}) correctly returns
a splitting of $z_j$ into two parts $z^{(d)}_j$ and $z_j^{(r)}$ and a set of indices $Q \subseteq \{j+1, \dots, i_{s-1}-1\}$ s.t. 
$$q_j = z^{(d)}_j + \sum_{\ell \in Q} R[\ell] = m_{j\, j} + \sum_{k \geq j+1} m_{k\, j} = \sum_{k \geq i_s} m_{k\, j}$$ 
where the first equality holds after the execution of  lines \ref{inner-for1-start}-\ref{inner-for1-end}, 
and the second equality holds because { by (i) $m_{k\, j} = 0$ for 
$k < j$}. 
We have  established the first equation   
of (\ref{prop1:eq}).
Moreover, the second equation of (\ref{prop1:eq}) also holds because by  the equality in (\ref{lemma9:eq2}), the result of the assignment in line \ref{inner-for1-end} 
and (by (i), with $j = i_{s-1}-1$)  $m_{i_{s-1}\, k} = 0$ for $i_s \leq k < i_{s-1}-1$, we get
\begin{eqnarray*}
p_{i_{s-1}-1} &=& 
z_{i_{s-1}-1}^{(r)} + z_{i_{s-1}-1}^{(d)}
+ \sum_{k > i_{s-1}-1} m_{i_{s-1}-1\, k}\\
&=& R[i_{s-1}-1] + \sum_{k \geq i_{s}} m_{i_{s-1}-1\, k},
\end{eqnarray*}
We now argue for the cases 
$j = i_{s-1}-2, i_{s-1}-3,  \dots, i_s.$ By induction we can assume that 
{ at the beginning of any  iteration of the {\bf for}-loop (lines \ref{int2.1-start}-\ref{int2.1-end}) with 
$i_s \leq j < i_{s-1}-1,$ we have that }
for each $i_{s-1}-1 \geq j' > j$ 
\begin{equation} \label{lemma4:indhyp1}
\sum_{k=i_s}^n m_{k\, j'} = q_{j'} \quad
\sum_{k=i_s}^n m_{j' \, k} = p_{j'} - R[j']
\end{equation}
and (if $s > 1$) by Corollary \ref{lemmata:prop4} for each $j' \geq i_{s-1}$ we have
\begin{equation} \label{lemma4:indhyp1-2-bis}
\sum_{k=i_s}^n m_{k\, j'} = q_{j'} \qquad \mbox{and} \qquad \sum_{k=i_s}^n m_{j'\, k} = p_{j'}
\end{equation}  
Moreover, by point 1.\ and 3.\ of Fact \ref{fact:p-q-z} we have  
$$z_j  = p_j  \quad  \mbox{ and } \quad \sum_{k = j}^n z_k  = \sum_{k=j}^n p_k \geq \sum_{k=j}^n q_k.$$
 From these, we have 
\begin{eqnarray}
q_j &\leq&  z_j + \sum_{k=j+1}^n z_k  -  \sum_{k=j+1}^n q_k  \\ \nonumber 
    &=&  z_j + \sum_{k=j+1}^n p_k - \sum_{k=j+1}^n q_k  \\ 
&=&   z_j +   \sum_{k=j+1}^{i_{s-1}-1} \left( \sum_{r=i_s}^n m_{k \, r} + R[k] \right) 
                    + \sum_{k = i_{s-1}}^n \sum_{r=i_s}^n m_{k\, r}  
                                    - \sum_{k=j+1}^n \sum_{r=i_s}^n m_{r\, k}  \label{lemma4:eqnarray3}\\
&=&       z_j  + \sum_{k=j+1}^{i_{s-1}-1} R[k]  \!\! + \! \sum_{k=j+1}^n \sum_{r=i_s}^n m_{k\, r} 
                  - \sum_{k=j+1}^n \sum_{r=i_s}^n m_{r\, k}  \label{lemma4:eqnarray4}\\
&=&   z_j + \! \sum_{k=j+1}^{i_{s-1}-1} R[k] \label{lemma4:eqnarray5}
\end{eqnarray}

\medskip
\noindent
where 
(\ref{lemma4:eqnarray3}) follows 
by using  
(\ref{lemma4:indhyp1}) and (\ref{lemma4:indhyp1-2-bis}); 
(\ref{lemma4:eqnarray4}) follows from (\ref{lemma4:eqnarray3}) by simple algebraic manipulations; {finally  
(\ref{lemma4:eqnarray5}) follows from (\ref{lemma4:eqnarray4}) because, 
by (i), at the end of iteration $j+1$, we have
$m_{\ell \, c} = 0$ if $\ell < j+1$ or $c < j+1$; hence 
$$\sum_{k=j+1}^n \sum_{r=i_s}^n m_{k\, r}  = \sum_{k=j+1}^n \sum_{r=j+1}^n m_{k\, r}$$ 
and 
 $$\sum_{k=j+1}^n \sum_{r=i_s}^n m_{r\, k}  = \sum_{k=j+1}^n \sum_{r=j+1}^n m_{r\, k}$$ and the equal terms and cancel out. }
 
 \smallskip
For each $k = j+1, \dots, i_{s-1}-1$ such that $R[k] \neq 0$ we have 
$R[k] = z^{(r)}_k \leq z_k \leq z_j$, where we are using the fact that for $\z = \p \wedge \q$ it holds that $z_1 \geq \cdots \geq z_n.$
 
Therefore $z_j, q_j$ and the values in $R[j+1\dots, i_{s-1}-1]$ satisfy the
 hypotheses of Lemma \ref{lemma:pezzettini}. Hence, 
the call in line 7 to  Algorithm \ref{algo:Lemma} (implementing the construction of Lemma \ref{lemma:pezzettini}) correctly returns
a splitting of $z_j$ into two parts $z^{(d)}_j$ and $z_j^{(r)}$ and a set of indices $Q \subseteq \{j+1, \dots, i_{s-1}-1\}$ s.t. 
$q_j = z^{(d)}_j + \sum_{\ell \in Q} R[\ell].$ Then we can use the same argument used in the first part of this proof 
(for the base case $j = i_{s-1}-1$) to show that  the first equation 
of (\ref{prop1:eq}) holds 
after lines \ref{inner-for1-start}-\ref{inner-for1-end}.

For $j' = j$, the second equation in (\ref{prop1:eq}) is guaranteed by the assignment in line \ref{inner-for1-end}.
Moreover, for $j' > j$ and $j' \not \in Q$ it holds 
since no entry $m_{j'\, k}$ or $R[j']$ is modified in lines  
\ref{inner-for1-start}-\ref{inner-for1-end}. Finally, for $j' > j$ and $j'  \in Q$ 
before the execution of lines \ref{inner-for1-start}-\ref{inner-for1-end} we had
$p_{j'} = R[j'] + \sum_{k \geq i_s} m_{j' \, k}$ with $R[j] > 0$ and $m_{j'\, j} = 0$
and after the execution of lines \ref{inner-for1-start}-\ref{inner-for1-end}
the values of $R[j']$ and  $m_{j' \, j}$ are swapped, hence the equality still holds. 
The proof of the second equality of (\ref{prop1:eq}) is now complete.
\end{IEEEproof}

\begin{corollary} \label{lemmata:prop2}
When algorithm \ref{algo:Algonew}  reaches line \ref{int2.1-if}, it holds that 
\begin{equation} \label{prop2:eq1-q}
\sum_{k \geq i_s} m_{k \, j} = q_j  \mbox{ and } \sum_{k < i_s} m_{k \, j} = 0 
\quad \mbox{ for } j \geq i_s 
\end{equation}
\begin{equation} \label{prop2:eq1-p}
\sum_{k \geq i_s - 1} m_{j \, k} = p_j  \mbox{ and } \sum_{k < i_s - 1} m_{j \, k} = 0
\quad \mbox{ for } j \geq  i_s 
\end{equation}
and (if $i_s \neq 1$, i.e., this is not the last iteration of the outermost {\bf for}-loop)
\begin{equation} \label{prop2:eq2}
\sum_{k} m_{k \, i_{s}-1} =  \sum_{k=i_s}^n p_k  - \sum_{k=i_s}^n q_k  = q_{i_s -1} - z_{i_s-1}.
\end{equation}
\end{corollary}

\begin{IEEEproof}
We will only prove (\ref{prop2:eq1-q}) and (\ref{prop2:eq1-p}) for $j = i_s, \dots, i_{s-1}-1.$ In fact, this is 
all we have to show if $s=1$. Furthermore, if $s >1,$ then in the previous iteration of the outermost {\bf for}-loop
the algorithm has reached line \ref{int2.2-if} and by 
Corollary \ref{lemmata:prop4} (\ref{prop2:eq1-q}) and (\ref{prop2:eq1-p}) also hold for 
each $j \geq i_{s-1},$ as desired.

By Lemma \ref{lemmata:prop1} when the algorithm reaches line \ref{if1} we have that
for each $j = i_s, \dots, i_{s-1}-1$ 
\begin{equation} \label{coro:recap} 
\sum_{k \geq i_s} m_{k\, j} = q_j \mbox{ and } \sum_{k \geq i_s}m_{j \, k} = p_j - R[j]
\end{equation}
Note that the entries of $M$  in columns $i_{s}, \dots, i_{s-1}-1$ will not be changed again and within the {\bf for}-loop 
in lines  \ref{int2.1-start}-\ref{int2.1-end} only the values $m_{i \, j}$ with $j = i_{s}, \dots i_{s-1}-1$ and $i =j, \dots, i_{s-1}-1$ 
may have been changed. Hence,  for each $j= i_s, \dots, i_{s-1}-1,$ it holds that 
$\sum_{k < i_s} m_{k \, j} = 0$ and $\sum_{k < i_s - 1} m_{j \, k} = 0$ as desired.

{ Moreover, by Lemma \ref{lemmata:prop1} (i) with $j = i_s$ it holds that $m_{\ell \, c} =  0,$ when  
$\ell < i_s$ or $c < i_s.$
%
%
Hence,  for each $j= i_s, \dots, i_{s-1}-1,$ it holds that 
$\sum_{k < i_s} m_{k \, j} = 0$ and $\sum_{k < i_s - 1} m_{j \, k} = 0$ as desired.}

Since the operations in lines \ref{if1}-\ref{int2.1-if} only change values in column $i_{s}-1$ of $M$ and in the vector $R$, 
the first equation in (\ref{coro:recap}) directly implies that (\ref{prop2:eq1-q}) holds for each $j = i_{s}, \dots, i_{s-1}.$

With the aim of proving (\ref{prop2:eq1-p}) let us first observe that if $i_s = 1$ (hence $s = k$ and the
algorithm is performing the last iteration 
of the outermost {\bf for}-loop) then by (\ref{prop2:eq1-q}) and (\ref{coro:recap}) we have
\begin{eqnarray} 
1 &=& \sum_{j=1}^n q_j \\
   &=& \sum_{j=1}^n \sum_{k=1}^n m_{k\, j} \\
 &=& \sum_{j=1}^n \sum_{k=1}^n m_{j\, k} \\
 &=&  \sum_{j=1}^n p_j - \sum_{j=1}^{i_{s-1}-1} R[j] \\
&=& 1 - \sum_{j=1}^{i_{s-1}-1} R[j]
\end{eqnarray}
and, since $R[j] \geq 0,$ it follows that $R[k] = 0,$ for each $j = i_s, \dots, i_{s-1}-1.$

Now, first assume that $i_s > 1$ hence $s < k$
From (\ref{coro:recap}) for each $j=i_s, \dots, i_{s-1}-1$ such that $R[j] = 0$ we immediately have that 
(\ref{prop2:eq1-p}) is also satisfied. Hence, this is the case for all $j=i_s, \dots, i_{s-1}-1$ 
when $s=k$ and $i_s=1.$
Moreover,  if there is some $j \in \{i_s, \dots, i_{s-1}-1\}$ 
such that $R[j] \neq 0$ (when $s < k$ and $i_s > 1,$), 
after the execution of line \ref{inversion-sums1-end},  
for each $j=i_s, \dots, i_{s-1}-1$ such that $R[j]$ was  $\neq 0$ 
we have $m_{j \, i_s-1} = R[j]$, hence,
$$\sum_{k \geq i_s -1} m_{j\, k} = m_{j\, i_s-1} + \sum_{k \geq j} m_{j\, k} = R[j] + p_j - R[j]$$
completing the proof of (\ref{prop2:eq1-p}). 

Finally, we prove (\ref{prop2:eq2}). By the assignments in line \ref{inversion-sums1-end} and the fact that this is the first time that values in 
column $i_{s}-1$ of $M$ are set to non-zero values, from point 5 of Fact \ref{fact:p-q-z} we get 
\begin{eqnarray*}
q_{i_s-1} - z_{i_s-1} &=& \sum_{k=i_s}^n p_k - \sum_{k=i_s}^n q_k \\
&=& \sum_{k=i_s}^{i_{s-1}-1} p_k + \sum_{k=i_{s-1}}^n p_k - \sum_{k=i_s}^n q_k \\
&=& \sum_{k=i_s}^{i_{s-1}-1} \sum_{\ell \geq i_s - 1} m_{k \, \ell}  + \sum_{k=i_{s-1}}^n \sum_{\ell \geq i_s} m_{k\, \ell}
                                - \sum_{k=i_s}^n \sum_{\ell \geq i_s} m_{\ell \, k}\\
&=& \sum_{k=i_s}^{i_{s-1}-1} m_{k\, i_{s}-1} + \sum_{k=i_s}^n \sum_{\ell \geq i_s} m_{k\, \ell} 
                  - \sum_{k=i_s}^n \sum_{\ell \geq i_s} m_{\ell \, k}\\
&=& \sum_{k=i_s}^{i_{s-1}-1} m_{k\, i_{s}-1}
\end{eqnarray*}
 that, together with the fact that at this point the only non-zero values in column $i_{s}-1$ of $M$ are in the rows
 $i_{s}, \dots, i_{s-1}-1,$ completes the proof of (\ref{prop2:eq2}).
\end{IEEEproof}

\begin{lemma} \label{lemmata:prop3}
{At the end of each iteration of the {\bf for}-loop of lines lines \ref{int2.2-start}-\ref{int2.2-end} in Algorithm  \ref{algo:Algonew} 
($s = 1, \dots, k,$ and $i_s\leq j <i_{s-1}$)}
we have {(i) $m_{\ell\, c} = 0$ for each $\ell, c$ such that $\min\{\ell, c\} < j$; 
and (ii)}
 for each $j' = j, \dots, i_{s-1}-1$,
$$ 
\sum_{k} m_{j' \, k} = p_{j'}  \qquad \mbox{ and } \qquad C[j'] + \sum_{k} m_{k\, j'} = q_{j'}.
$$ 
\end{lemma}
\begin{IEEEproof}
The proof 
can be 
easily obtained by proceeding like in Lemma \ref{lemmata:prop3} (swapping the roles of 
rows and columns of $\mathbf{M}$ and $\p$ and $\q$).
\end{IEEEproof}

\begin{corollary} \label{lemmata:prop4}
When algorithm \ref{algo:Algonew}  reaches line \ref{int2.2-if}, it holds that 
$$ 
\sum_{k \geq i_s} m_{j \, k} = p_j  \mbox{ and } \sum_{k < i_s} m_{j\, k} =  0 
\quad \mbox{ for } j \geq i_s 
$$ 
$$ 
\sum_{k \geq i_s-1} m_{k \, j} = q_j  \mbox{ and } \sum_{k < i_s-1} m_{k\, j} =  0
\quad \mbox{ for } j \geq  i_s 
$$ 
and {if $i_s \neq 1$ (the outermost {\bf for}-loop is not in last iteration)}
$$ 
\sum_{k} m_{i_{s}-1 \, k} =  \sum_{k=i_s}^n q_k  - \sum_{k=i_s}^n p_k  = p_{i_s -1} - z_{i_s-1}.
$$ 
\end{corollary}
\begin{IEEEproof}
The proof 
can be 
easily obtained by proceeding like in Corollary \ref{lemmata:prop2} (swapping the roles of 
rows and columns of $\mathbf{M}$ and of $\p$ and $\q$).
\end{IEEEproof}

\subsection{{How  \textbf{Algorithm 1} works: An informal   description of its functioning}}\label{informal}
Given  the inversion points $i_0, i_1, \dots, i_k$ of the two probability distributions $\p$ and  $\q$, as defined in Definition \ref{defi:inversions},
 for each $s=1, \dots, k$  let us call   the list of integers $L^s = \{i_{s-1}-1, i_{s-1}-2 , \dots, i_{s}\}$ 
(listed in decreasing order)  a $\p$-\emph{segment},  or a $\q$-\emph{segment}, according to whether $s$ is odd or even. 
For each $i$ belonging to a $\p$-segment we have 
$$\sum_{j=i}^n z_j = \sum_{j = i}^n p_j \geq \sum_{j = i}^n q_j.$$  
For each $i$ belonging to a $\q$-segment we have 
$$\sum_{j=i}^n z_j = \sum_{j = i}^n q_j \geq \sum_{j = i}^n p_j.$$

\smallskip
\textbf{Algorithm 1}  proceeds by filling entries of the matrix $M$ with non-zero values.
 Other possible actions of the 
algorithm consist in moving probabilities from one entry of $M$ to a neighboring one. 
The reasons of this  moving will become clear as the description of the algorithm unfolds.

At any point during the execution of the algorithm, 
we say that a column $i$ is satisfied if the 
sum of the entries on column $i$ is equal  to $q_i$. Analogously, we say that a row $i$ is satisfied if the 
sum of the entries on row $i$ is equal  to $p_i$.
Obviously, the goal of the algorithm is to satisfy all rows and columns.
Line \ref{ini:3} makes sure that the first value of $j$ in line \ref{int2.1-start} is in a $\p$-segment.

Line \ref{ini:3} makes sure that the first value of $j$ in line \ref{int2.1-start} is in a $\p$-segment.
For each $j = n, \dots, 1,$ with $j$ in a $\p$-segment, the algorithm maintains the following invariants:

\begin{itemize}
\item[1-p] all rows $j' > j$ are satisfied
\item[2-p] all columns $j' > j$ are satisfied
\item[3-p] the non-zero entries $M[j', j]$ for $j' > j$ in the same $\p$-segment, satisfy $M[j',j]+M[j',j'] = z_{j'}$
\end{itemize}

The main steps of the algorithm when $j$ is in a $\p$-segment amount  to: 
\begin{itemize}
\item[]
\begin{itemize}
\item[Step1-p:] Put $z_j$ in $M[j, j]$. By the assumption that $j$ is part of a $\p$-segment, we have that this assignment satisfies also row $j$. 
However, this assignment  might create an 
excess on column $j$, (i.e., the sum of the elements on column $j$ could be greater than the  value of
$q_j$) since by the invariants 1-p and 2-p and the assigned value to the  entry $M[j,j]$ we have that the sum of all the entries filled so far equals
to  
$$\sum_{j' \geq j} z_j' = \sum_{j' \geq j} p_j' \geq \sum_{j' \geq j} q_j',$$ and the entries on the columns $j' > j$ satisfy exactly $q_{j'}$, that is, 
sums up to $q_{j'}$.
\item[Step2-p:] If there is an excess on column $j$, we  adjust it by applying Lemma 3. Indeed, by Lemma 3 
we can select entries $M[j',j]$ that together with part of $M[j,j] = z_j$ sum up exactly to 
$q_j.$ The remaining part of $z_j = M[j,j]$ and each of the non-selected entries on column $j$ are kept on their same row but 
moved to column $j-1$. In the pseudocode of {Algorithm 1}, this operation is simulated by using the auxiliary array $R[\cdot]$.
Notice that by this operation we are maintaining invariants $1$-p, $2$-p, $3$-p for $j \leftarrow j+1.$
\end{itemize}
\end{itemize} 
Step1-p and Step2-p are repeated as long as $j$ is part of a $\p$-segment. 

When $j$ becomes part of a $\q$-segment, the roles pf $\p$ and $\q$ are inverted, namely, we have that the following invariants hold:

\begin{itemize}
\item[1-q] all rows $j' > j$ are satisfied
\item[2-q] all columns $j' > j$ are satisfied
\item[3-q] the non-zero entries $M[j, j']$ for $j' > j$ in the same $\q$-segment, satisfy $M[j,j']+M[j',j'] = z_{j'}$
\end{itemize}

From now on and as long as $j$ is part of a $\q$-segment the main steps of the algorithm amount  to: 
\begin{itemize}
\item[]
\begin{itemize}
\item[Step1-q:] Put $z_j$ in $M[j, j]$. By the assumption that $j$ is part of a $\q$-segment, we have that this assignment satisfies also column $j$. 
Again, this assignment   might create an 
excess on row $j$, since by the invariants 1-q and 2-q and the entry $M[j,j]$ we have that the sum of all the entries filled so far is equal to  
$$\sum_{j' \geq j} z_j' = \sum_{j' \geq j} q_j' \geq \sum_{j' \geq j} p_j'$$ 
and the entries on the rows $j' > j$ satisfy exactly $p_{j'}.$
\item[Step2-q:] If there is an excess on row $j$, by Lemma 3 we can select entries $M[j,j']$ that together with part of $M[j,j] = z_j$ sum up exactly to 
$p_j.$ The remaining part of $z_j = M[j,j]$ and the non-selected entries  are kept on the same column but are moved up to 
row $j-1$. In the pseudocode of {Algorithm 1}, this operation is simulated by using the auxiliary array $C[\cdot]$.
Notice that by this operation we are maintaining invariants $1$-q, $2$-q, and $3$-q for $j \leftarrow j+1.$
\end{itemize}
\end{itemize}
 
Again these steps are repeated as long as $j$ is part of a $\q$-segment. When $j$ becomes part of a $\p$-segment again, 
 we will have that once more invariants $1$-p, $2$-p, and $3$-p are satisfied. Then,  
 the algorithm resorts to repeat steps Step1-p and Step2-p as long as $j$ is part of a $\p$-segment, 
and so on and so forth switching between  $\p$-segments and $\q$-segments, until all the rows and columns are satisfied.
 
From the above description it is easy to see that all the values used to fill in entries 
of the matrix are created by splitting into two parts some 
element $z_j$. 
This key property of the algorithm  implies the bound on the entropy of $M$ being at most $H(\z) + 1.$

We note here that (for efficiency reasons) in the pseudocode of Algorithm \ref{algo:Algonew} instead of moving values from one column to the next one (Step2-p) or from 
one row to the next one (Step2-q), the arrays $R$ and $C$ are used,
where $R[j']$ plays the role of $M[j',j]$, in invariant $3$-p above,  and 
$C[j']$ plays the role of $M[j,j']$, in invariant $3$-q above.


\section{Extending the result to more distributions}\label{sec:more}

In this section we will show how the algorithm {\sc  Min-Entropy-Joint-Distribution} can be used to attain a $\log k$ additive approximation 
for the problem of minimising the entropy of the joint distribution of $k$ input distributions, for any $k \geq 2.$

In what follows, for the ease of the description, we shall assume that $k= 2^{\kappa}$ for some integer $\kappa \geq 1,$ 
i.e., $k$  is a power of $2$.  A pictorial perspective on the algorithm's behaviour is to imagine that the input distributions are in the 
leaves of a complete binary tree with $k = 2^{\kappa}$ leaves. Each internal node $\nu$ of the tree contains the joint distribution of the 
distributions in the leaves of the subtree rooted at $\nu$. Such a distribution is computed by applying 
the algorithm {\sc  Min-Entropy-Joint-Distribution} to the distributions in the children of $\nu$. 

The algorithm builds such a tree starting from the leaves. Thus, 
the  joint distribution of all the input distributions will be given by the distribution computed at the root of the tree. 

We will denote by $m^{(i-j)}$ the non-zero components of the distribution that our algorithm builds as joint distribution of 
$\p^{(i)}, \p^{(i+1)}, \dots, \p^{(j)}.$ Algorithm \ref{algo:multidistribution} shows the pseudocode for our procedure.

The vector $Indices^{(i-j)}$ is used to record for each component $\m^{(i-j)}[w]$ the indices of the component of the joint probability distribution of $\p^{(i)}, \dots, \p^{(j)}$ which coincides with  $\m^{(i-j)}[w]$. Therefore, if after the execution of line 17, 
for $w = 1, \dots, |\m^{(i-j)}|$, we have $Indices^{(i-j)}[w] = \langle s_i[w], s_{i+1}[w], \dots, s_j[w]  \rangle$ it means that
setting  $M^{(i-j)}[s_i[w], s_{i+1}[w], \dots, s_j[w]] \leftarrow \m^{(i-j)}[w]$ and setting the remaining components of $M^{(i-j)}$ to zero, 
the array $M^{(i-j)}$ is a joint distribution matrix for $\p^{(i)}, \dots, p^{(j)}$ whose non-zero components are equal to the components
of $\m^{(i-j)}.$ Hence, in particular, we have that $H(M^{(i-j)}) = H(\m^{(i-j)}).$

The algorithm explicitly uses this correspondence only for the final array $M^{(1-k)}$ representing the joint distribution of all input distributions. 

Based on the above discussion the correctness of the algorithm can be easily verified. In the rest of this section we will prove that the entropy of the 
joint distribution output by the algorithm guarantees additive $\log k$ approximation.  

We will prepare some definitions and lemmas which will be key tools for proving the approximation guarantee of our algorithm. The proof of these technical lemmas is deferred to the next section.

\begin{algorithm}[ht!] 
\small
{\sc  K-Min-Entropy-Joint-Distribution}($\p^{(1)}, \p^{(2)} , \dots, \p^{(k)}$)\\ 
{\bf Input:} prob.\ distributions $\p^{(1)}, \p^{(2)} , \dots, \p^{(k)},$ with $k = 2^{\kappa}$\\
{\bf Output:} A $k$-dimensional array $\mathbf{M} = [m_{i_1, i_2, \dots, i_k}]$ s.t.\ $\sum_{i_1, \dots, i_{j-1}, i_{j+1}, \dots, i_k} 
m_{i_1, \dots, i_{j-1}, t,  i_{j+1}, \dots, i_k} = p^{(j)}_t$ for each $j = 1, \dots, k$ and each $t$.
\begin{algorithmic}[1]
\FOR{$i = 1$ {\bf to} $k$}
\FOR{$j = 1$ {\bf to} $n$}
\STATE{{\bf set} $\m^{(i-i)}[j] = \p^{(i)}_j$ and $Indices^{(i-i)}[j] = \langle j \rangle$} \label{ini:1}
\COMMENT{$Indices^{(i-i)}[j]$ is a vector of indices} 
\ENDFOR
\ENDFOR
\STATE{{\bf for} $i=1,\dots, k$  {\bf permute the components of} $\m^{(i-i)}$ and $Indices^{(i-i)}$ using the permutation that sorts $\m^{(i-i)}$ in non-increasing order} \label{ini:2}
\FOR{$\ell = 1$ {\bf to} $\kappa$} \label{main:1}
\STATE{$i \leftarrow 1, \, j \leftarrow 2^{\ell}$}
\WHILE{$j \leq k$}  \label{int2.1-start}
\STATE{$j_1 \leftarrow i+2^{\ell-1}-1, \, j_2 = j_1+1$}
\STATE{$M \leftarrow$ {\sc  Min-Entropy-Joint-Distribution}$(\m^{(i-j_1)}, \m^{(j_2-j)})$}   
\STATE{$w \leftarrow 1$}
\FOR{$s=1$ {\bf to} $|\m^{(i-j_1)}|$}
\FOR{$t=1$ {\bf to} $|\m^{(j_2-j)}|$}
\IF{$M[s,t] \neq 0$}
\STATE{$\m^{(i-j)}[w] \leftarrow M[s,t]$}
\STATE{$Indices^{(i-j)}[w] \leftarrow Indices^{(i-j_1)}[s] \odot Indices^{(i-j_1)}[t]$} 
\COMMENT{$\odot$ denotes the concatenation of vectors} 
\STATE{$w \leftarrow w+1$}
\ENDIF
\ENDFOR
\ENDFOR
\STATE{{\bf permute the components of} $\m^{(i-j)}$ and $Indices^{(i-j)}$ using the permutation that sorts $\m^{(i-j)}$ in non-increasing order} 
\ENDWHILE 
\STATE{$i \leftarrow j+1, \, j \leftarrow i+2^{\ell}-1$}
\ENDFOR  
\STATE{{\bf set} $M[i_1,i_2,\dots, i_k] = 0$ for each $i_1, i_2, \dots, i_k.$}
\FOR{$j = 1$ {\bf to} $|\m^{(1-k)}|$}
\STATE{$M[Indices^{(1-k)}[j]] \leftarrow \m^{(1-k)}[j]$}
\ENDFOR
\STATE{{\bf return} $M$}
\end{algorithmic}
\caption{The Min Entropy Joint Distribution Algorithm for $k>2$ distributions}
\label{algo:multidistribution}
\end{algorithm}

Let us define the following:

\begin{definition}
For any $\p = (p_1, \dots, p_n) \in \PP_n$ we denote by $\half(\p)$ the distribution  
$(\frac{p_1}{2}, \frac{p_1}{2}, \frac{p_2}{2}, \frac{p_2}{2}, \dots, \frac{p_n}{2}, \frac{p_n}{2})$ obtained by splitting each component of 
$\p$ into two identical halves.

For any $i \geq 2,$ let us also define $\half^{(i)}(\p) = \half(\half^{(i-1)}(\p)),$ where $\half^{(1)}(\p) = \half(\p)$ and   $\half^{(0)}(\p) = \p.$
\end{definition}

\medskip

We will employ the following two technical lemmas whose proofs are in the next section.

\begin{lemma} \label{monotone-morphism}
For any $\p \preceq \q$ we have also $\half(\p) \preceq \half(\q)$ 
\end{lemma}


\begin{lemma} \label{key-morphism-power}
For any pair of distributions $\p, \q \in \PP_n.$ and any $i \geq 0$, It holds that 
$$\half^{(i)}(\p \wedge \q) \preceq \half^{(i)}(\p) \wedge \half^{(i)}(\q).$$ 
\end{lemma}

\begin{theorem} \label{k-joint:main}
For each $\ell = 0, 1, \dots \kappa$ and $s = 0, 1, 2, \dots, k/2^{\ell}-1$ let $i = i(\ell, s) =  s \cdot 2^{\ell} + 1$ and 
$j = j(\ell, s) = (s+1) \cdot 2^{\ell} = i + 2^{\ell} -1.$ Then, we have 
$$\half^{(\ell)}(\p^{(i)} \wedge \p^{(i+1)} \wedge \cdots \wedge \p^{(j)}) \preceq \m^{(i-j)}.$$
\end{theorem}
\begin{IEEEproof}
The proof is by induction on $\ell.$ The base case follows by definition of the operator $\half^{(\ell)}$ and the 
fact that the algorithm sets $\m^{(i-i)} = \p^{(i)},$ for each $i$ hence in particular  $\m^{(i-i)} = \p^{(i)} = \half^{(0)}(\p^{(i)}),$ which 
proves the desired inequality.

We now prove the induction step. Let $\ell > 0.$ It is enough to consider only the case $s = 0,$ since the other cases are perfectly analogous. 

Therefore, $i = 1$ and $j = 2^{\ell}.$ Using the notation employed in the pseudocode, let $j_1 = 2^{\ell-1}, \, j_2 = 2^{\ell-1}+1.$ 
By induction hypothesis we can assume that 
\begin{equation} \label{eq:inductionhypothesis-1}
\half^{(\ell-1)}(\p^{(i)} \wedge \p^{(i+1)} \wedge \cdots \wedge \p^{(j_1)}) \preceq \m^{(i-j_1)}
\end{equation}
\begin{equation} \label{eq:inductionhypothesis-2}
\half^{(\ell-1)}(\p^{(j_2)} \wedge \p^{(j_2+1)} \wedge \cdots \wedge \p^{(j)}) \preceq \m^{(j_2-j)}.
\end{equation}

It follows that
\begin{eqnarray}
\half^{(\ell)}\left(\bigwedge_{\iota = i}^j \p^{(\iota)}  \right) &=& 
    \half^{(\ell)}\left( \left( \bigwedge_{\iota = i}^{j_1} \p^{(\iota)} \right) 
                               \wedge \left( \bigwedge_{\iota = j_2}^j \p^{(\iota)}    \right)  \right)   \label{eq:half1}\\
    &=& \half \left( \half^{(\ell-1)}\left( \left( \bigwedge_{\iota = i}^{j_1} \p^{(\iota)} \right) 
                                                          \wedge \left( \bigwedge_{\iota = j_2}^j \p^{(\iota)}    \right)  \right) \right) \label{eq:half2}\\
    &\preceq& \half\left( \half^{(\ell-1)} \left( \bigwedge_{\iota = i}^{j_1} \p^{(\iota)}  \right)    \wedge
                             \half^{(\ell-1)} \left( \bigwedge_{\iota = j_2}^{j} \p^{(\iota)}  \right)   \right)   \label{eq:half3}\\
    &\preceq& \half\left( \m^{(i-j_1)} \wedge \m^{(j_2-j)} \right) \label{eq:half4}\\
    &\preceq& \m^{(i-j)}  \label{eq:half5}
\end{eqnarray}
where 
\begin{itemize}
\item (\ref{eq:half2}) follows from (\ref{eq:half1}) by the definition of the operator $\half$;
\item (\ref{eq:half3}) follows from (\ref{eq:half2}) by Lemma \ref{key-morphism-power};
\item (\ref{eq:half4}) follows from (\ref{eq:half3}) by the induction hypotheses (\ref{eq:inductionhypothesis-1})-(\ref{eq:inductionhypothesis-2}) ;
\item (\ref{eq:half5}) follows from (\ref{eq:half4}) by observing that the components of $\m^{(i-j)}$ 
coincide with the components of the array $M$ output by algorithm 
{\sc Min-Entropy-Joint-Distribution} executed on the distributions $\m^{(i-j_1)}$ and $\m^{(j_2-j)}.$ 
Let $\z = \m^{(i-j_1)} \wedge \m^{(j_2-j)}$ and $|\z|$ denote the number of components of $\z.$
By the analysis presented in the previous section we have that we can partition the components of $M$ (equivalently, the components
of $\m^{(i-j)}$) into subsets 
$M_1, M_2, \dots, M_{|\z|}$ such that 
\begin{itemize}
\item $1 \leq |M_i| \leq 2$
\item for each $i = 1, \dots, |\z|,$ it holds that  $\sum_{x \in M_i} x = z_i$;
\end{itemize}
Therefore---assuming, w.l.o.g., that the components of $\m^{(i-j)}$ are reordered such that those in $M_i$ immediately precede those in 
$M_{i+1}$---we have $\half(\z) = \m^{(i-j)} P$ where $P = [p_{i\, j}]$ is a doubly stochastic matrix defined by 
$$p_{i\,j} = \begin{cases}
\frac{1}{2} &  \mbox{if ($i$ is odd {\bf and} } j \in \{i, i+1\} \mbox{) {\bf or}  ($i$ is even {\bf and} } j \in \{i, i-1\}); \\
0 & otherwise
\end{cases} 
$$
\end{itemize}
from which it follows that $\half(\z) \preceq \m^{(i-j)}$ yielding \ref{eq:half5}.
\end{IEEEproof}

An immediate consequence of the last theorem is the following
\begin{corollary}
For any  $k$ probability distributions $\p^{(1)}, \dots, \p^{(k)}$ let $M$ be the joint distribution of $\p^{(1)}, \dots, \p^{(k)}$
 output by algorithm 
{\sc K-Min-Entropy-Joint-Distribution}. Then,  
$$H(M) \leq H(\p^{(1)} \wedge \p^{(2)} \wedge \cdots \p^{(k)}) + \lceil \log k \rceil$$ 
\end{corollary}
\begin{proof}
Let $k$ be a power of $2$. Otherwise repeat some of the probability distribution until there are  $k' = 2^{\lceil \log k \rceil}$ of them. 
By Theorem \ref{k-joint:main} we have 
$$\half^{(\lceil \log k \rceil )}(\p^{(1)} \wedge \p^{(2)} \wedge \cdots \p^{(k)}) =  
\half^{(\log k'  )}(\p^{(1)} \wedge \p^{(2)} \wedge \cdots \p^{(k')})  \preceq \m^{(1-k)}.$$

Therefore, by the Schur-concavity of the entropy we have 
$$H(M) = H(\m^{(1-k)}) \leq H(\half^{(\lceil \log k \rceil )}(\p^{(1)} \wedge \p^{(2)} \wedge \cdots \p^{(k)})) 
= H(\p^{(1)} \wedge \p^{(2)} \wedge \cdots \p^{(k)}) + \lceil \log k \rceil,$$
where the last equality follows by the simple observation that 
for any probability distribution $\x$ and integer $i \geq 0$ we have $H(\half^{(i)}(\x)) = H(\x) + i.$
\end{proof}

\medskip

We also have the following lower bound which, together with the previous corollary implies that 
our algorithm guarantees an additive $\log k$ approximation for the problem of 
computing the joint distribution of minimum entropy of $k$ input distributions. 

\begin{lemma}
Fix $k$ distributions $\p^{(1)}, \p^{(2)},  \cdots,  \p^{(k)}$. For any $M$ being a joint distribution of 
$\p^{(1)}, \p^{(2)},  \cdots, \p^{(k)}$ it holds that 
$$H(M) \geq H(\p^{(1)} \wedge \p^{(2)} \wedge \cdots \p^{(k)})$$
\end{lemma}
\begin{proof}
For each $i = 1, \dots, k,$ the distribution $\p^{(I)}$ is an aggregation of $M$, hence $M \preceq \p^{(i)}.$

By definition of the greatest lower bound operator $\wedge$ for any distribution $\x$ such that for each $i$ it holds that $\x \prec \p^{(i)}$ we have 
$\x \preceq \p^{(1)} \wedge \p^{(2)} \wedge \cdots \p^{(k)}$. Therefore, in particular we have
$M \preceq \p^{(1)} \wedge \p^{(2)} \wedge \cdots \p^{(k)},$ which, by the Schur concavity of the entropy 
gives the desired result. 
\end{proof}

 Summarising we have shown the following
 
 \begin{theorem}
Let $\p^{(1)}, \dots, \p^{(m)} \in \PP_n$. Let $M^*$ be a joint distribution of $\p^{(1)}, \dots, \p^{(m)}$ of minimum entropy among all 
the joint distribution of $\p^{(1)}, \dots, \p^{(m)}.$
Let $M$ be the joint distribution of $\p^{(1)}, \dots, \p^{(m)}$ output by our algorithm. Then, 
$$H(M) \leq H(M^*) + \lceil \log(m) \rceil.$$
Hence, our (polynomial) algorithm provides an additive $\log(m)$ approximation.
\end{theorem}
  
\subsection{The proofs of the two technical lemmas about the operator $\half$}

\medskip
\noindent
{\bf Lemma \ref{monotone-morphism}.}
{\em For any $\p \preceq \q$ we have also $\half(\p) \preceq \half(\q)$.} 

\begin{proof}
It is easy to see that assuming $\p$ and $\q$ rearranged in order to have 
$p_1 \geq p_2 \geq \dots \geq p_n$ and $q_1 \geq q_2 \geq \dots \geq q_n$ we also have
$\half(\p)_1 \geq \half(\p)_2 \geq \dots \geq \half(\p)_{2n}$ and $\half(\q)_1 \geq \half(\q)_2 \geq \dots \geq \half(\q)_{2n}.$

By assumption we also have that for each $j = 1, \dots, n$ it holds that $\sum_{i=1}^j p_i \leq  \sum_{i=1}^j p_i.$

Therefore, for each $j=1, \dots 2n$ it holds that 
$$\sum_{i=1}^j \half(\p)_i = \frac{1}{2} \sum_{i=1}^{\lceil j/2 \rceil} p_i + \frac{1}{2} \sum_{i=1}^{\lfloor j/2 \rfloor} p_i \leq
\frac{1}{2} \sum_{i=1}^{\lceil j/2 \rceil} q_i + \frac{1}{2}\sum_{i=1}^{\lfloor j/2 \rfloor} q_i  = \sum_{i=1}^j \half(\q)_i,$$
proving that $\half(\p) \preceq \half(\q).$
\end{proof}

\medskip

\begin{lemma} \label{key-morphism}
For any pair of distributions $\p, \q \in \PP_n.$ It holds that 
$$\half(\p \wedge \q) \preceq \half(\p) \wedge \half(\q).$$ 
\end{lemma}
\begin{proof}
By the previous lemma we have that 
$$\half(\p \wedge \q) \preceq \half(\p) \qquad \mbox{and} \qquad \half(\p \wedge \q) \preceq \half(\q)$$
Then, by the property of the operator $\wedge$ which gives the greatest lower bound we have the desired result.
\end{proof}

On the basis of this Lemma we can extend the result to "powers" of the operator $\half$ and have 
our Lemma \ref{key-morphism-power}. 

\medskip
\noindent
{\bf Lemma \ref{key-morphism-power}.} {\em 
For any pair of distributions $\p, \q \in \PP_n.$ and any $i \geq 0$, It holds that 
$$\half^{(i)}(\p \wedge \q) \preceq \half^{(i)}(\p) \wedge \half^{(i)}(\q).$$ 
}
\begin{proof}
We argue by induction on $i$. The base case $i=1$ is given by the previous lemma. Then,
for any $i > 1$ 

$$\half^{(i)}(\p \wedge \q) = \half(\half^{(i-1)}(\p \wedge \q)) \preceq 
\half(\half^{(i-1)}(\p)  \wedge \half^{(i-1)}(\q)) \preceq \half(\half^{(i-1)}(\p)) \wedge  \half(\half^{(i-1)}(\p))$$
from which the desired result immediately follows. The first $\preceq$-inequality follows by induction hypothesis and the
second inequality by Lemma \ref{key-morphism}.
\end{proof}


\newpage

\newpage
\appendix
\subsection{A numerical example}
Fix $n = 13$ and let $$\p = (0.35, 0.095, 0.09, 0.09, 0.09, 0.09, 0.08, 0.06, 0.035, 0.015, 0.003, 0.001, 0.001)$$ and 
$$\q = (0.15, 0.15, 0.145, 0.145, 0.14, 0.13, 0.05, 0.03, 0.03, 0.027, 0.002, 0.0005, 0.0005),$$ be the two probability 
distribution for which we are seeking a joint probability of minimum entropy.
By Fact \ref{glb} we have
$$\z = \p \land \q = (0.15, 0.15, 0.145, 0.145, 0.125, 0.09, 0.08, 0.055, 0.03, 0.025, 0.003, 0.001, 0.001).$$ 

By Definition \ref{defi:inversions} we have that the inversion points are $i_0 = 14, i_1 = 11, i_2 = 9, i_3 = 6, i_4 = 1.$

The resulting joint distribution produced by \textbf{Algorithm 1} is given by the following matrix $\mathbf{M} = [m_{i\,j}]$,
satisfying the property that   $\sum_{j} m_{i\, j} = p_i$ and $\sum_{i} m_{i\, j} = q_j.$

\begin{Large}
$${\mathbf{\scalebox{0.85}{M}}} = \left (\begin{smallmatrix}
0.15 & 0.145 & 0.055 & 0 & 0 & 0 & 0 & 0 & 0 & 0 & 0 & 0 & 0 \\
0 & 0.005 & 0 & 0.09 & 0 & 0 & 0 & 0 & 0 & 0 & 0 & 0 & 0 \\
0 & 0 & 0.09 & 0 & 0 & 0 & 0 & 0 & 0 & 0 & 0 & 0 & 0 \\
0 & 0 & 0 & 0.055 & 0.035 & 0 & 0 & 0 & 0 & 0 & 0 & 0 & 0 \\
0 & 0 & 0 & 0 & 0.09 & 0 & 0 & 0 & 0 & 0 & 0 & 0 & 0 \\
0 & 0 & 0 & 0 & 0.015 & 0.075 & 0 & 0 & 0 & 0 & 0 & 0 & 0 \\
0 & 0 & 0 & 0 & 0 & 0.055 & 0.025 & 0 & 0 & 0 & 0 & 0 & 0 \\
0 & 0 & 0 & 0 & 0 & 0 & 0.025 & 0.03 & 0.005 & 0 & 0 & 0 & 0 \\
0 & 0 & 0 & 0 & 0 & 0 & 0 & 0 & 0.025 & 0.01 & 0 & 0 & 0 \\
0 & 0 & 0 & 0 & 0 & 0 & 0 & 0 & 0 & 0.015 & 0 & 0 & 0 \\
0 & 0 & 0 & 0 & 0 & 0 & 0 & 0 & 0 & 0.002 & 0.001 & 0 & 0 \\
0 & 0 & 0 & 0 & 0 & 0 & 0 & 0 & 0 & 0 & 0.0005 & 0.0005 & 0 \\
0 & 0 & 0 & 0 & 0 & 0 & 0 & 0 & 0 & 0 & 0.0005 & 0 & 0.0005 
\end{smallmatrix}\right )
$$
\end{Large}

Notice that by construction
\begin{itemize}
\item for the submatrix $\mathbf{M}^{(i_1)} = [m_{i\,j}]_{i_1 \leq i \leq i_0-1, \, i_1-1 \leq j \leq i_0-1}$ we have that each row $i$ contains at most 
two elements and the sum of the elements on the row equals $z_i$
\item for the submatrix $\mathbf{M}^{(i_2)} = [m_{i\,j}]_{i_2-1 \leq i \leq i_1-1, \, i_2 \leq j \leq i_1-1}$ we have that each column $i$  contains at most 
two elements and the sum of the elements on the column equals $z_i$
\item for the submatrix $\mathbf{M}^{(i_3)} = [m_{i\,j}]_{i_3 \leq i \leq i_2-1, \, i_3-1 \leq j \leq i_2-1}$ we have that each row $i$  contains at most 
two elements and the sum of the elements on the row equals $z_i$
\item for the submatrix $\mathbf{M}^{(i_4)} = [m_{i\,j}]_{i_4 \leq i \leq i_3-1, \, i_4 \leq j \leq i_1-1}$ we have that each column $i$ c contains at most 
two elements and the sum of the elements on the column $z_i$
\end{itemize}

Notice that these four sub-matrices cover all the non-zero entries of $\mathbf{M}$.
This easily shows that the entries of the matrix $\mathbf{M}$ are obtained by splitting into at most two pieces the components of $\z$,
 implying the desired bound 
$H(\mathbf{M}) \leq H(\z) +1.$

\end{document}